\documentclass{amsart}
\usepackage{amsmath}
\usepackage{amsthm}

\newtheorem{theorem}{Theorem}[section]
\newtheorem{lemma}[theorem]{Lemma}

\theoremstyle{definition}

\newtheorem{assumption}[theorem]{Assumption}

\theoremstyle{remark}
\newtheorem{remark}[theorem]{Remark}

\newcommand{\mysection}[1]{\section{#1}
\setcounter{equation}{0}}

\newcommand{\norm}[1]{\left\Vert#1\right\Vert}
\newcommand{\abs}[1]{\left\vert#1\right\vert}
\renewcommand{\epsilon}{\varepsilon}

\begin{document}
\title{on a nonlinear recurrent relation}
\thanks{This material is based upon work supported by the National
Science Foundation under agreement No. DMS-0111298. Any options,
findings and conclusions or recommendations expressed in this
material are those of the authors and do not necessarily reflect
the views of the National Science Foundation.}

\author[D. Li]{Dong Li}
\address[D. Li]
{School of Mathematics, Institute for Advanced Study,
Einstein Drive, Princeton, NJ 08540, USA}
\email{dongli@math.ias.edu}



\begin{abstract}
We study the limiting behavior for the solutions of a nonlinear recurrent relation
which arises from the study of Navier-Stokes equations \cite{SL06}.
Some stability theorems are also shown concerning a related class of linear
recurrent relations.
\end{abstract}

\maketitle

\mysection{Introduction and main theorems}
In this paper we consider the following nonlinear recurrent relation,
\begin{align}  \label{07:recurr:main}
\Lambda_{p}(x)  = \frac 1 {p} \sum_{\substack{p_1+p_2=p\\p_1,p_2\ge 1}} f(\gamma) \Lambda_{p_1} (x)
 \Lambda_{p_2} (x), \quad \gamma=\frac {p_1} {p}, \; p> 1
\end{align}
where $\Lambda_1=x \in \mathbb R$ is a parameter, and $f: [0,1]\rightarrow \mathbb R$ has
integral $1$. This quadratic recurrent relation arises from our recent study of complex blow ups of the $3D$
Navier-Stokes system \cite{SL06}. There $f(\gamma)$ takes the special form
$f(\gamma) = 6\gamma^2 -10 \gamma +4$, and we need to show that for each initial value $\Lambda_1=x$,
there exists $R(x)$ such that
$$
\Lambda_p(x) = R(x)^p ( 1 + \delta_p ),
$$
and $\delta_p\rightarrow 0$ as $p\rightarrow \infty$. The main object of this paper is to
prove this claim for a general class of functions $f$ including our special function.

It is clear that $\Lambda_p(\lambda x) = \lambda^p \Lambda_p(x)$ for any $\lambda\ne 0$. Therefore
it suffices for us to show that there exists some $x^*$ such that
$ \Lambda_p(x^*) = 1 + \delta_p\rightarrow 1$ as $p \rightarrow \infty$. The limiting value
of $\Lambda_p(x^*)$ is $1$ since the function $f$ has integral $1$ over $[0,1]$. Clearly
$\Lambda_p(x) = (x/x^*)^p\Lambda_p(x^*)=R(x)^p (1+\delta_p)$ if we define $R(x) =x/x^*$.

We now state our conditions on the function $f$.

\begin{assumption} \label{07assump_main}
Let $\tilde f(\gamma) = f(\gamma)+f(1-\gamma)$ satisfy the following:
\begin{enumerate}
\item The integral of $\tilde f$ over $[0,1]$ is $2$.
\item $\tilde f$ is a polynomial, i.e.
\begin{align} \label{07:assump:1}
 \tilde f(\gamma) = \sum_{n=0}^N a_n \gamma^n
\end{align}
\item The equation in $\alpha$:
\begin{align} \label{07:assump:2}
\sum_{n=1}^N \frac {na_n} {n+2} \cdot \frac 1 {n+\alpha} =1,
\end{align}
has $N$ distinct roots $z_1, \cdots, z_N$, in $\mathbb C \cap \{ z: -1<Re(z) < 0 \}$.
Denote
$$
\alpha_f = 1 - \max_{k} Re(z_k).
$$
\end{enumerate}
\end{assumption}
One more technical condition, which is needed for the inductive stage of our proof,
and can be justified with easy numerics, will be stated in section 5 (see remark
\ref{07:rem:technical}).
We now state our main theorem.

\begin{theorem}[Main theorem]
Let assumption \ref{07assump_main} and the technical condition in
section 5 hold. Then there exists $x^*$ and constant $C>0$, such
that the solution $\Lambda_p(x)$ to the recurrent relation
\eqref{07:recurr:main} satisfies
$$
\abs{ \Lambda_p(x^*) -1 } \le \frac {C} { p^{\alpha_f}}, \quad \forall\ p>1.
$$
Consequently for any $x$, there exists a unique number $R=x/x^*$, such that
$$
\Lambda_p (x) = R^p \left( 1 + \frac {C_p} { p^{\alpha_f}} \right), \quad
\forall\ p>1,
$$
where $C_p$ does not depend on x and is uniformly bounded in $p$.
\end{theorem}
\begin{remark}
In the special case $f(\gamma)=6\gamma^2 - 10 \gamma +4$ \cite{SL06}, we have
$\tilde f(\gamma)=12\gamma^2 -12\gamma+4$,
and $z_{1,2} = \frac {-1\pm\sqrt{15} i} 2$ are the roots of \eqref{07:assump:2}.
Clearly the assumptions are satisfied with $\alpha_f = 3/2$. Our theorem says
that in this case
$$
\Lambda_p(x) = R^p \left( 1 + O(p^{-3/2}) \right).
$$
More examples can be easily constructed.
\end{remark}
The rest of this paper is organized as follows. In section 2 we derive the
linear system and give the main arguments. Section 3 includes the technical estimates
needed to derive the linear system. Section 4 contains some theorems on the behavior
of the solutions to the linearized system. Section 5 is devoted to the estimates of
the nonlinear term and the proof of the main theorem. Some elementary estimates with
explicit control of constants are deferred to the appendix.

\mysection{Preliminary analysis and linearization}
Since the limiting value of $\Lambda_p(x^*)$ is $1$, we linearize
\eqref{07:recurr:main} around $1$ and
this gives:
\begin{align*}
\Lambda_{p} (x) = \frac 1 {p} \sum_{p_1=1}^{p-1} f(\gamma)
+ \frac 1 {p} \sum_{p_1=1}^{p-1} \tilde f(\gamma)
 (\Lambda_{p_1}(x) -1)
\\ + \frac 1 {p} \sum_{p_1+p_2=p } f(\gamma) (\Lambda_{p_1}(x)-1)
  (\Lambda_{p_2}(x)-1),
\end{align*}
where $\tilde f(\gamma) = f(\gamma) +f(1-\gamma)$. Our main goal is to show
that there exists $x^*$ such that $\Lambda_p(x^*)\rightarrow 1$.
This motivates us to define $a_{p}$ such that
$\Lambda_{p}(a_{p}) =1$. When $p$ tends to infinity $a_p$ is a good approximation of $x^*$.
Now let us write
$$
\frac {a_{p}} {a_{p-1}} = 1 + \frac {\xi_p} {p^3}.
$$
The scaling $p^{-3}$ here is not intuitively obvious and in fact is not optimal since
as we shall see, $\abs{\xi_p}\le Const \cdot p^{-(\alpha_f-1)}$ as $p$ tends to infinity.
In terms of $\xi_p$ we have
\begin{align*}
\left( 1 + \frac {\xi_p} {p^3} \right)^{-p} &=  \frac 1 p
\sum_{p_1=1}^{p-1}  f(\gamma) +\frac 1 p \sum_{p_1=1}^{p-2} \tilde
f(\gamma)
 (\Lambda_{p_1}(a_{p-1}) -1) \\
 & \quad
 + \frac 1 p \sum_{p_1+p_2=p } f(\gamma) (\Lambda_{p_1}(a_{p-1})-1)
(  \Lambda_{p_2}(a_{p-1})-1).
\end{align*}
Or in a better form,
\begin{align} \label{eq07_27_10}
&\underbrace{p\left(1-\frac 1 p \sum_{p_1=1}^{p-1} f(\gamma) \right)}_{R_p^{(1)}}
+\underbrace{p\left(\left(1+\frac {\xi_p} {p^3} \right)^{-p} -1\right)}_{R_p^{(2)}}
- \underbrace{
\sum_{p_1=1}^{p-2} \tilde f(\gamma) \left( \prod_{p_1<q<p} \left(1+\frac {\xi_q} {q^3} \right)^{p_1} -1\right)
}_{Q_p} \nonumber\\
=&
\underbrace{
\sum_{p_1+ p_2 = p} f(\gamma)
\left(\prod_{p_1 <q <p} ( 1+ \frac {\xi_q} {q^3})^{p_1}-1 \right)
\cdot
\left( \prod_{p_2<q<p} ( 1+ \frac {\xi_q} {q^3} )^{p_2}-1 \right)
}_{N_p}.
\end{align}
We shall derive a linear system for $\xi_p$ from the above expression. First observe that
for small values of $p_1$, the summand in $Q_p$ are of order $1$ and therefore the main order of
$Q_p$ is not linear in $\xi_q$, $q\le p-1$. However as we shall show in section 3, we have
\begin{align*}
p( Q_p - Q_{p-1}) =\frac p {p-1} \xi_{p-1} - \frac 1 p \sum_{q=2}^{p-1} \xi_q G(q/p)+ O(\log p/p),
\end{align*}
where $G(\gamma)=\int_0^1 \tilde f^\prime(t\gamma) t^2 dt$ and $O(\log p /p)$ denotes terms of higher
order of smallness in $p$. As for $R_p^{(1)}$, since by assumption $f$ is a polynomial, it follows (by
direct summation) that
\begin{align*}
p(R_p^{(1)} - R_{p-1}^{(1)}) = O\left( \frac 1{p} \right).
\end{align*}
Similarly
\begin{align*}
p(R_p^{(2)} - R_{p-1}^{(2)}) = - \xi_p + \frac p {p-1} \xi_{p-1} + O\left( \frac 1{p^2} \right).
\end{align*}
The term $N_p$ will be estimated in section 5 and it is of higher order of smallness in $p$. Put all
these considerations together, we see that \eqref{eq07_27_10} is equivalent to
\begin{align} \label{eq27_10_19}
\xi_p = \frac 1 p \sum_{q=2}^{p-1} G(q/p) \xi_q + h_p,
\end{align}
where $h_p$ are of higher order in $p$.

The main difficulty of estimating $\xi_p$ by using \eqref{eq27_10_19} is "loss of control of constants".
To explain this problem, take our special function $f(\gamma)=6\gamma^2-10\gamma+4$ and this gives
$G(\gamma) = 6\gamma-4$. It is enough to consider the system
\begin{align*}
\xi_p = \frac 1 p \sum_{q=2}^{p-1} G(q/p) \xi_q = \frac 1 p \sum_{q=2}^{p-1} (6q/p-4) \xi_q.
\end{align*}
Suppose we want to prove by induction that
$\abs{\xi_r} \le A r^{-\alpha}$ for some $\alpha\ge 0$ and $A>0$. Then at step $p$ we will get
\begin{align*}
\abs{\xi_p} p^{\alpha} \le C_p: = A \frac 1 p \sum_{q=2}^{p-1} \abs{6q/p-4} \left(\frac q p\right)^{-\alpha}.
\end{align*}
As $p\rightarrow\infty$, clearly
\begin{align*}
C_p \rightarrow A \int_0^1 \abs{6\gamma-4} \gamma^{-\alpha} d\gamma \ge A \int_0^1 \abs{6\gamma-4}d\gamma >A.
\end{align*}
In other words we will not be able to justify the inductive hypothesis for $p$ sufficiently large.
For general function $\tilde f$ one can show easily that the
integral of $G$ over $[0,1]$ is $-1$ (assuming the integral of $\tilde f$ is 2). This implies that
$\int_0^1 \abs{G(\gamma)} d\gamma \ge 1$ and therefore this "loss of control of constants" problem is
generic.\footnote{Even if $G(\gamma)\le 0$ and $\int_0^1 G=-1$, due to the nonlinear corrections $h_p$ in
\eqref{eq27_10_19}, we still have "loss of control of constants".}

To solve this problem, we will first prove in section 4 a stability theorem concerning the linear system
\eqref{eq27_10_19}. And instead of inducting on $\xi_p$, we shall induct on $h_p$. The stability theorem
in section 4 gives us bounds on $\xi_p$ by using the (inductively assumed) bounds on $h_p$. Since $h_{p+1}$ is bounded by
quadratic functions of all $\xi_{q}$, $q\le p$, the bounds on $\xi_{q}$ then produce a strong decay estimate on $h_{p+1}$
(see lemma \ref{lem_07_13_20}). By using a slightly weaker induction hypothesis on $h_p$ (relative to the strong decay
estimate),
we can justify our inductive bound at step $p+1$ at the sacrifice of assuming $p$ to be
sufficiently large. We are able to close our argument because of the genuine nonlinear nature of $h_p$.
\mysection{The estimates of $Q_p-Q_{p-1}$}
In this section we give the technical estimates of $Q_p-Q_{p-1}$.
By definition of $Q_p$, we have
\begin{align*}
 Q_p - Q_{p-1}
=& \sum_{p_1=1}^{p-2} \tilde f(\gamma) \cdot
\left( \prod_{p_1<q <p} \left( 1+ \frac {\xi_q} {q^3} \right)^{p_1}
-\prod_{p_1<q <p-1} \left( 1+ \frac {\xi_q} {q^3} \right)^{p_1} \right) \\
& +\sum_{p_1=1}^{p-3} ( \tilde f(\gamma) -\tilde f(\gamma') ) \cdot
\left( \prod_{p_1<q <p-1} \left( 1+ \frac {\xi_q} {q^3} \right)^{p_1} -1 \right)\\
=& \text{(I)}+\text{(II)}.
\end{align*}
We shall show that in the main order of magnitude, we have
\begin{align*}
\text{(I)} \approx \frac {\xi_{p-1}} {p-1},
\end{align*}
and
\begin{align*}
\text{(II)} \approx -\frac 1 {p^2} \sum_{q=2}^{p-1} \xi_q G(q/p).
\end{align*}
Throughout this section we make the following ansatz on $\xi_q$, $q<p$:
\begin{assumption} \label{assumption_07_400}
Let $A_1,A_2, A_3$ be positive constants. Let $N_0\ge \max\{3, A_1\}$ be
a positive integer, and $p>N_0$ such that:
\begin{enumerate}
\item $\abs{\xi_q} \le A_1$, for any $N_0\le q<p$.
\item $\abs{\xi_q} \le A_3$, for any $1\le q\le N_0$.
\item $\abs{ \prod_{p_1<q \le N_0} \left( 1 + \frac {\xi_q} {q^3} \right)^{p_1}} \le A_2$,
for any $1\le p_1 <N_0$.
\end{enumerate}
\end{assumption}
Denote $G(\gamma)= \int_0^1 \tilde f^\prime(\gamma t) t^2 dt$, we have the following main lemma.
\begin{lemma}[Main lemma] \label{Lem_07_main_11}
Let assumption \ref{assumption_07_400} hold.
Then there exists a constant
$C=C \left(\norm{\tilde f}_\infty, \norm{\tilde f^\prime}_\infty, \norm{\tilde f^{\prime\prime}}_\infty,
A_1,A_2,A_3,N_0 \right)$, such that
\begin{align*}
\abs{ p(Q_p -Q_{p-1})  - \frac p {p-1} \xi_{p-1}+\frac 1 p \sum_{q=2}^{p-1} G(q/p) \xi_q  } \le C \frac {\log p} {p}.
\end{align*}
\end{lemma}
\begin{proof}
Recall that $Q_p - Q_{p-1} = \text{(I)}+\text{(II)}$ and
we will estimate (I) and (II) separately. First note that
$\int_0^1 \tilde f(\gamma) \gamma d\gamma =1$. Using this fact we have
\begin{align*}
\text{(I)}=& \sum_{p_1=1}^{p-2} \tilde f(\gamma) \left(
\prod_{p_1< q < p}\left(1 + \frac {\xi_q} {q^3} \right)^{p_1} -1\right)
\cdot \left( 1 -\left(1+ \frac {\xi_{p-1}}{(p-1)^3} \right)^{-p_1} \right) \\
&\;+\sum_{p_1=1}^{p-2} \tilde f(\gamma) \left( 1 - \left(1+ \frac {\xi_{p-1}}{(p-1)^3}\right)^{-p_1}
 - \frac {p_1 \xi_{p-1}} {(p-1)^3} \right) \\
 &\; + \left( \sum_{p_1=1}^{p-2} \tilde f(\gamma) \frac {p_1} {(p-1)^2} - \int_0^1 \tilde f(\gamma)
 \gamma d\gamma \right) \frac {\xi_{p-1}} {p-1}+  \frac {\xi_{p-1}} {p-1}\\
=& e_p^{(1)} + e_p^{(2)} + e_p^{(3)}+\frac {\xi_{p-1}} {p-1}.
 \end{align*}
\texttt{Estimate of $e_p^{(1)}$}: clearly
\begin{align*}
\abs{ e_p^{(1)}} \le &\norm{\tilde f}_\infty \sum_{p_1=1}^{p-2}
\abs{ \prod_{p_1<q<p} \left(1+\frac {\xi_q} {q^3} \right)^{p_1}-1}
\cdot \frac {p_1 A_1} {(p-1)^3} \\
\le &\norm{\tilde f}_\infty \sum_{N_0 \le p_1 \le p-2} \frac {A_1} {p_1} \cdot \frac {p_1 A_1} {(p-1)^3}
+ \\
&\;+ \norm{\tilde f}_\infty\sum_{1\le p_1 <N_0} \left( 1 + A_2 (1+\frac A{N_0} )\right)
\cdot \frac {p_1 A_1} {(p-1)^3} \\
\le &\norm{\tilde f}_\infty\left(
\frac {A_1^2} {(p-1)^2} +
\frac {(1+A_2 (1+\frac {A_1} {N_0})) {A_1} N_0^2}{2(p-1)^3} \right).
\end{align*}

\texttt{Estimate of $e_p^{(2)}$}: first it is rather easy to show that
\begin{align*}
\abs{ 1 - \left( 1 + \frac {\xi_{p-1}} {(p-1)^3} \right)^{-p_1}
} \le \frac {p_1 A_1} {(p-1)^3}.
\end{align*}
\begin{align*}
\abs{ \left( 1+ \frac {\xi_{p-1}} {(p-1)^3} \right)^{p_1} -1
- \frac {p_1 \xi_{p-1}} {(p-1)^3}}
\le \frac {3p_1^2 A_1^2} {2(p-1)^6},
\end{align*}
and also
\begin{align*}
\abs{ \left( 1+ \frac {\xi_{p-1}} {(p-1)^3} \right)^{-p_1} -1
+ \frac {p_1 \xi_{p-1}} {(p-1)^3}}
\le \frac {3p_1^2 A_1^2} {2(p-1)^6}.
\end{align*}
Now it follows easily that
\begin{align*}
\abs{e_p^{(2)}} \le & \norm{\tilde f}_\infty \sum_{p_1=1}^{p-2} \frac {3p_1^2 A_1^2} {2(p-1)^6} \\
\le &\norm{\tilde f}_\infty \frac {A_1^2} {2(p-1)^3}.
\end{align*}

\texttt{Estimate of $e_p^{(3)}$}:
This estimate actually does not require assumption \ref{assumption_07_400}.
The only assumption needed here is that $p\ge 4$. Although this is a standard estimate,
for the purpose of explicit control of constants, we give the full details here. Clearly
\begin{align*}
& \abs{ \sum_{p_1=1}^{p-2} \tilde f(\gamma) \cdot \frac {p_1} {p-1} \cdot
    \frac 1 {p-1} - \int_0^1 \tilde f(\gamma)\gamma d\gamma
} \\
\le & \sum_{p_1=1}^{p-2} \int_{\frac {p_1-1} {p-1}}^{\frac {p_1} {p-1}}
\abs{ \tilde f\left( \frac {p_1} p \right)\cdot \frac {p_1} {p-1}
- \tilde f(\gamma) \gamma}d\gamma + \int_{\frac {p-2} {p-1}}^1 \abs{\tilde f(\gamma)} \gamma d\gamma\\
\le & \sum_{p_1=1}^{p-2} \int_{\frac {p_1-1} {p-1}}^{\frac {p_1} {p-1}}
\norm{\tilde f^\prime}_\infty \abs{\gamma-\frac {p_1} p} \frac {p_1} {p-1} d\gamma
+ \norm{\tilde f}_\infty \cdot \frac {p-2} {2(p-1)^2} +\frac 1 {p-1} \norm{\tilde f}_\infty \\
\le & \frac 1 {4(p-1)} \norm{\tilde f^\prime}_\infty
+ \frac 3 {2(p-1)} \norm{\tilde f}_\infty.
\end{align*}

The estimates of (II) are similar. Write
 \begin{align*}
 \text{(II)} = & \sum_{p_1=1}^{p-3} \left( \tilde f(\gamma) - \tilde f(\gamma^\prime) + \tilde f^\prime(\gamma)
 \cdot \frac {\gamma} {p-1} \right) \left( \prod_{p_1<q < p-1} \left( 1+ \frac {\xi_q} {q^3} \right)^{p_1} -1 \right)\\
 &\;+ \sum_{p_1=1}^{p-3} \left( - \tilde f^\prime(\gamma) \cdot \frac {\gamma} {p-1} \right)
 \cdot  \left( \left( 1+ \frac {\xi_q} {q^3} \right)^{p_1} -1 - p_1 \sum_{p_1<q<p} \frac {\xi_q} {q^3} \right) \\
 &\;+\left( \sum_{p_1=1}^{p-3} \left( - \tilde f^\prime(\gamma) \cdot \frac {\gamma} {p-1} \right)
 \cdot \left( p_1 \sum_{p_1<q<p} \frac {\xi_q} {q^3} \right) + \frac 1 {p^2} \sum_{q=2}^{p-2} \xi_q G(q/p) \right)\\
 &\; - \frac 1 {p^2} \sum_{q=2}^{p-2} \xi_q G(q/p) \\
=& e_p^{(4)} + e_p^{(5)} + e_p^{(6)} - \frac 1 {p^2} \sum_{q=2}^{p-2} \xi_q G(q/p)
\end{align*}
\texttt{Estimate of $e_p^{(4)}$ and $e_p^{(5)}$}: by Taylor expansion on $\tilde f$ we have
\begin{align*}
\abs{\tilde f(\gamma) - \tilde f(\gamma^\prime) + \tilde f^\prime(\gamma)
 \cdot \frac {\gamma} {p-1} } \le \frac 1 2 \norm{\tilde f^{\prime\prime}}_{\infty}
\cdot \frac {\gamma^2} {(p-1)^2}.
\end{align*}
Then by lemma \ref{lem_07_16_34} in the appendix,
\begin{align*}
\abs{e_p^{(4)}}\le  \norm{\tilde f^{\prime\prime}}_{\infty}\left(
\frac {A_1} {4p^2} + \frac {N_0^3} {6(p-1)^2 p^2} \left( 1 +A_2(1+ \frac {A_1} {N_0}) \right)
\right).
\end{align*}
Similarly the estimate of $e_p^{(5)}$ follows from lemma \ref{lem_07_16_39} in the
appendix:
\begin{align*}
\abs{e_p^{(5)}}
\le  \norm{\tilde f^\prime}_\infty \left(\frac {A^2 \log (p-3) } {4p(p-1)} +
\frac 1 {p(p-1)} \left(
 \frac {1+A_2} 2 N_0^2 + \frac {(3A_2+1) {A_1} +3 A_3} 6 N_0 \right) \right).
\end{align*}

\texttt{Estimate of $e_p^{(6)}$}: First note that
 \begin{align*}
& \abs{
 \frac 1 {p-1}\sum_{1\le p_1 <q} \tilde f^\prime (\gamma) \gamma^2
- \int_0^{ \frac {q} p} \tilde f^\prime(\gamma) \gamma^2 d\gamma
} \\
\le & \sum_{p_1=1}^{q-1} \int_{\frac {p_1-1} p}^{\frac {p_1} p}
\abs{ \frac p {p-1} \tilde f^\prime\left( \frac {p_1} p\right) \cdot
 \left(\frac {p_1} p\right)^2 - \tilde f^\prime(\gamma) \gamma^2 } d\gamma
+ \int_{\frac {q-1}p}^{\frac q p} \abs{\tilde f(\gamma)} \gamma^2 d\gamma\\
\le & \sum_{p_1=1}^{q-1} \left( \norm{\tilde f^\prime}_\infty \frac {p_1^2} {p^3 (p-1)}
+ \left( \norm{\tilde f^{\prime\prime}}_\infty + 2 \norm{\tilde f^\prime}_\infty \right)
\frac {p_1} {p^3} \right) + \norm{\tilde f^\prime}_\infty \frac {q^2} {p^3} \\
\le & \norm{\tilde f^\prime}_\infty \left( \frac {q^3} {3p^4} + \frac {2q^2} {p^3} \right)
+ \norm{ \tilde f^{\prime\prime} }_\infty \frac {q^2} {2p^3}.
 \end{align*}
Then
\begin{align*}
&\sum_{p_1=1}^{p-3} \tilde f^\prime(\gamma) \cdot \frac {\gamma} {p-1}
 \sum_{p_1<q<p-1} \frac {p_1 \xi_q} {q^3} \\
=& \frac 1 {p^2} \sum_{q=2}^{p-2} \xi_q \sum_{1\le p_1<q} \frac {p^3} {q^3}
\frac 1 {p-1} \tilde f^\prime(\gamma) \cdot \gamma^2
\end{align*}
Now it follows easily that
\begin{align*}
\abs{e_p^{(6)}}= &\abs{ \frac 1 {p^2} \sum_{q=2}^{p-2} \xi_q \sum_{1\le p_1<q} \frac {p^3} {q^3}
\frac 1 {p-1} \tilde f^\prime(\gamma) \cdot \gamma^2
- \frac 1 {p^2} \sum_{q=2}^{p-2} \xi_q G(q/p) } \\
\le & \frac 1 {p^2} \sum_{q=2}^{p-2} \abs{\xi_q} \left(
\norm{\tilde f^\prime}_\infty \left( \frac 1 {3p} + \frac 2 q \right)
+ \norm{\tilde f^{\prime\prime}}_\infty \frac 1 {2q} \right).
\end{align*}
The lemma is proved.
\end{proof}

\mysection{Behavior of solutions to the linear system}
In this section we shall study the linear system defined by
$$
\xi_p = \frac 1 p \sum_{q=2}^{p-1} G(q/p) \xi_q + h_p.
$$
As was already mentioned in section 2, this linear system is obtained from the
nonlinear relation \eqref{eq07_27_10}. We shall now free ourselves from the
constraint that $h_p$ is a quadratic function of all $\xi_q$, $q\le p$. Instead we
assume here that $h_p$ is known. Under some conditions on the function $G$ and the sequence $h_p$, we
show that $ \xi_p \approx const \cdot p^{\sigma}$, where $\sigma=\sigma(G)$ is
to be defined later in this section. We begin with a simple lemma.
\begin{lemma} \label{lem24.2.M}
Let $N\ge 1$. Assume
$\{C_i\}_{1\le i\le N}$, $\{\alpha_i\}_{1\le i\le N}$,
$\{ \sigma_i \}_{1\le i\le N}$ are complex numbers such that:
\begin{enumerate}
\item $\alpha_i \ne \alpha_j$, $\sigma_i \ne \sigma_j$, if $i\ne j$.

\item 
$
\sum_{i=1}^N \frac {C_i}{\sigma_k +\alpha_i+1} = 1,
\quad\forall\ 1\le k\le N.
$

\item $Re(\sigma_k)\le 0$, $\forall\ 1\le k\le N$.
\end{enumerate}
Suppose $M_p \in {\mathbb R}^{N\times N}$, $p\ge 2$
is a sequence of matrices defined by
$$
M_p^{(k,j)} =
\left(1 - \frac 1 p \right)^{\alpha_k+1} \delta_{kj}
+C_j \cdot \left(1 -\frac 1 p \right)^{\alpha_j+1}
\cdot \frac 1 p.
$$
Then there exists an a constant $K>0$ which depends only
on $\alpha_i$,$\sigma_i$, $C_i$, such that
\begin{align*}
\Bigl\| M_p M_{p-1}\cdots M_{q_0+1}M_{q_0}
\Bigr\|_2 \le K,
\end{align*}
for any $q_0 \ge 2$, $p\ge 3$,
where $\norm{\cdot}_2$ denotes the matrix spectral norm.
\end{lemma}
\begin{proof}
We have
\begin{align*}
M_p^{(k,j)} &= \left(1 - \frac{\alpha_k+1} p \right)
\delta_{kj} + \frac {C_j} {p}
+ O\left( \frac 1 {p^2} \right) \\
& = \delta_{kj} + \frac{C_j - (\alpha_k+1)\delta_{kj}} p
+ O\left( \frac 1 {p^2} \right).
\end{align*}
Consider the matrix $\tilde M$ defined by
$$
\tilde M^{(k,j)} = C_j -(\alpha_k+1)\delta_{kj}.
$$
It is not difficult to show that the equation for
eigenvalues of $\tilde M$ $\det{( \tilde M -\lambda I)} =0$ is equivalent to:
$$
\sum_{i=1}^N \frac {C_i} {\lambda+\alpha_i+1} =1.
$$
By our assumption $\tilde M$ has $N$ distinct eigenvalues $\sigma_1,\cdots, \sigma_N$.
Let $D$ be the diagonal matrix such that $D_{ii}=\sigma_i$.
It follows easily that there exist $S,S^{-1}\in \mathbb C^{N\times N}$, such that
$$
M_p= S \left( I +  \frac 1 p D+ O\left( \frac 1 {p^2} \right) \right) S^{-1}.
$$
Since $ \max_{1\le k\le N} Re(\sigma_k) \le 0$, we have
$$
\norm{ I + \frac 1 p D +O\left( \frac 1 {p^2} \right) }_2 \le 1 + \frac {C_1} {p^2},
$$
where $C_1$ is a constant independent of $p$. Now the theorem follows by using
the submultiplicative property of the matrix spectral norm.
\end{proof}
In what follows, we shall assume that $G(\cdot)$ is a finite sum of
(generalized) monomials, i.e.
\begin{align} \label{eq23.11.G1}
G(\gamma) = \sum_{i=1}^N C_i \gamma^{\alpha_i},
\end{align} where
$\alpha_i$, $C_i$ are complex numbers, in particular, note
that $\alpha_i$ need not be integers. Here for $0<\gamma\le 1$ and $\alpha=a+bi$,
$\gamma^{\alpha}$ is simply defined as $\exp\{(a+bi) \log \gamma \}$ where
$\log \gamma$ is real-valued.
The main assumption on $G$ is the
following:\\
Assume that the equation in $\lambda$
\begin{align} \label{eq23.11.G2}
\sum_{i=1}^N \frac {C_i} {\lambda+\alpha_i +1} =1,
\end{align}
has $N$ distinct roots in $\mathbb C$. \\
Denote the $N$ distinct roots by $\sigma_1, \cdots, \sigma_N$,
and define
\begin{align} \label{eq23.11.G3}
\sigma(G) = \max_{1\le i\le N} Re(\sigma_i).
\end{align}
$\sigma$ will be called the characteristic exponent of $G$.
We prove the following theorem concerning a linear
recurrent system generated by $G$:
\begin{theorem}  \label{24.22.thm}
Assume that $G$ satisfies \eqref{eq23.11.G1} and
\eqref{eq23.11.G2}.
Consider the linear recurrent system defined by:
$$
\xi_p = \frac 1 p \sum_{q=2}^{p-1} G(q/p) \xi_q,\quad p\ge 3,
$$
with $\xi_2$ as a parameter. Then there exists a constant
$C$ depending only on $G$ such that
\begin{align}   \label{eq23.36.xi}
\abs{\xi_p} \le C\cdot p^{\sigma}
\cdot \abs{\xi_2}, \quad \forall\ p>1,
\end{align}
where the characteristic exponent
$\sigma=\sigma(G)$ is defined in \eqref{eq23.11.G3}.
\end{theorem}
\begin{proof}
We begin with a simple observation. Suppose
$Re(\sigma_{i_0}) = \sigma(G)$ and consider
$$
\tilde G(\gamma) = \sum_{i=1}^N C_i
 \gamma^{\alpha_i -\sigma_{i_0} }.
$$
Then $\tilde G(\gamma)$ also satisfies the assumption
\eqref{eq23.11.G2} with $\tilde \alpha_i=\alpha_i-\sigma_{i_0}$.
Also $\sigma(\tilde G)=0$,
$\tilde \xi_p = p^{-\sigma_{i_0}} \xi_p$, where $\tilde \xi_p$ is
generated by $\tilde G$. With this simple observation,
it suffices for us to prove the theorem assuming
$\sigma(G)=0$.
Assume this is the case and define the moments of $\xi_p$ by:
$$
B^{(k)}_p = \sum_{q=2}^p \xi_q q^{\alpha_k},
$$
where $\alpha_k$ is defined in the definition of $G$
(see \eqref{eq23.11.G1}). Then obviously we have
\begin{align} \label{eq24.2.52}
\xi_p = \sum_{k=1}^N \frac {C_k} {p^{\alpha_k+1}} B^{(k)}_{p-1}.
\end{align}
The recurrent formula for $B^{(k)}_p$ follows easily:
$$
B^{(k)}_p = B^{(k)}_{p-1} + p^{\alpha_k}\sum_{j=1}^N
\frac {C_j}{p^{\alpha_j+1}} B^{(j)}_{p-1}.
$$
Now if our bound on $\xi_p$ \eqref{eq23.36.xi} is correct, then
heuristically $B^{(k)}_p$ grows as $p^{Re(\alpha_k)+1}$.
This motivates us to define the scaled variables
$$
\tilde B^{(k)}_p = {p^{-\alpha_k-1}}{B^{(k)}_p}.
$$
For $\tilde B^{(k)}_p$ we have the recurrent relation:
$$
\tilde B^{(k)}_p = \sum_{j=1}^N M_p^{(k,j)}
\tilde B^{(j)}_{p-1},
$$
where the matrix $M_p^{(k,j)}$ is given by:
$$
M_p^{(k,j)} =
\left(1 - \frac 1 p \right)^{\alpha_k+1} \delta_{kj}
+C_j \cdot \left(1 -\frac 1 p \right)^{\alpha_j+1}
\cdot \frac 1 {p}.
$$
Note that the matrix $M_p$ is the same as in lemma
\ref{lem24.2.M}. Now denote the vector
$\tilde B_p =(\tilde B^{(1)}_p, \cdots,
\tilde B^{(N)}_p)^T$. Then we have
\begin{align*}
\norm{\tilde B_p}_2 &= \Bigl\| \left( M_p
M_{p-1}\cdots M_3M_2 \right)
\tilde B_1 \Bigr\|_2 \\
& \le  \Bigl\|  M_p
M_{p-1}\cdots M_3M_2
\Bigr\|_2\norm{\tilde B_1}_2.
\end{align*}
By lemma \ref{lem24.2.M} we have for some constant
$C$ depending only on $G$ such that
$$
\Bigl\|  M_p M_{p-1}\cdots M_3M_2 \Bigr\|_2 \le C.
$$
This immediately implies that $\tilde B_p$ is uniformly bounded
for all $p>1$. The desired bound on $\xi_p$ then follows by using
\eqref{eq24.2.52}. Our theorem is proved.
\end{proof}
The next lemma can be regarded as an inhomogeneous version of
theorem \ref{24.22.thm}. It states that the solution to
the linear recurrent system generated by $G$ is stable under sufficiently
small perturbations.
\begin{lemma} \label{lem_07_13_13}
 Assume $G$ satisfies \eqref{eq23.11.G1} and
\eqref{eq23.11.G2}, with the characteristic exponent
$\sigma=\sigma(G)$ defined in \eqref{eq23.11.G3}.
Consider the linear recurrent system defined by:
$$
\xi_p = h_p + \frac 1 p \sum_{q=2}^{p-1} G(q/p) \xi_q,\quad p\ge 3,
$$
where $\xi_2$ is a parameter, and assume that for some positive
constant $C_1$ and $\epsilon$, the sequence $h_p$ satisfies
$$
\abs{ h_p } \le {C_1}\cdot  {p^{-\epsilon+\sigma}}, \quad\forall\ p\ge 2.
$$
Then there exists a constant
$K$ independent of $p$, such that
\begin{align}   \label{eq24.22.xi}
\abs{\xi_p} \le KC_1\cdot p^{\sigma}
\cdot \abs{\xi_2}, \quad \forall\ p\ge 2.
\end{align}
\end{lemma}
\begin{proof}
Without loss of generality assume $\sigma(G)=0$ (see the beginning
of the proof of theorem \ref{24.22.thm}). Define $B_p^{(k)}$,
$\tilde B_p^{(k)}$, $\tilde B_p$ and $M_p^{(k,j)}$ as in the proof of theorem
\ref{24.22.thm}. Clearly then
\begin{align}   \label{eq.24.23.25}
\xi_p = h_p + \sum_{k=1}^N C_k \left( 1 - \frac 1 p \right)^{\alpha_k+1}
 \tilde B_{p-1}^{(k)}.
\end{align}
Denote $g_p = p^{-1}h_p (1,\cdots,1)^T$. Then for $\tilde B_p$ we have
\begin{align*}
\tilde B_p &= M_p \tilde B_{p-1} + g_p \\
& = M_p(M_{p-1} \tilde B_{p-2}+ g_{p-1}) + g_p = \cdots \\
& = \sum_{p_1=3}^p \Bigl( M_{p} M_{p-1} \cdots M_{p_1+1} \Bigr) g_{p_1}
+ \Bigl( M_p M_{p-1} \cdots M_3 \Bigr) \tilde B_2.
\end{align*}
By our assumption on $h_{p_1}$, we have
$$
\norm{ g_{p_1} }_2 \le \frac {\sqrt N C_1}{p_1^{1+\epsilon}},\quad
\forall\ 3\le p_1\le p.
$$
By lemma \ref{lem24.2.M}, we have for some constant $\tilde K>0$,
$$
\norm{ \Bigl( M_{p} M_{p-1} \cdots M_{p_1+1} \Bigr) }_2 \le \tilde K, \forall \, p_1,p\ge 2.
$$
It follows easily that
$$
\norm{\tilde B_p}_2 \le \tilde K \sum_{p_1=3}^p \frac {\sqrt NC_1} {p_1^{1+\epsilon}}
+\tilde K C_2 \le KC_1,
$$
where $C_2$ is a constant independent of $p$. Now use
\eqref{eq.24.23.25} to conclude the proof of the lemma. The lemma is proved.
\end{proof}
\mysection{Estimate of the nonlinear term and proof of the main theorem}
In this section we shall estimate the nonlinear term $N_p$. We have
the following lemma.
\begin{lemma} \label{lem_07_13_20}
Assume $0<\sigma<1$.
Let $(\xi_q)_{q\ge 2}$ be a sequence of numbers such that
\begin{align*}
\abs{ \prod_{p_1<q<N_0} \left( 1 + \frac {\xi_q} {q^3} \right)^{p_1} }
\le C_1, \quad\forall\ 1\le p_1< N_0.
\end{align*}
and
\begin{align*}
\abs{\xi_q} \le C_2 {q^{-\sigma}}, \quad \forall\  N_0<q<p,
\end{align*}
where $N_0\ge \max\{ C_2 ^{\frac 1 {1+\sigma}},
(4C_2)^{\frac 1 {3+\sigma}} \}$.
Let $p\ge 2N_0$ and $N_p$ be defined by
\begin{align*}
N_p = \sum_{p_1+ p_2 = p} f(\gamma)
\left(\prod_{p_1 <q <p} ( 1+ \frac {\xi_q} {q^3})^{p_1}-1 \right)
\cdot
\left( \prod_{p_2<q<p} ( 1+ \frac {\xi_q} {q^3} )^{p_2}-1 \right).
\end{align*}
Then we have
\begin{align*}
\abs{N_p} \le
\norm{\tilde f}_\infty \left(
\frac {(1+C_1+C_2) 2N_0^2 C_2 } {(p-N_0)^{2+\sigma}}
+ \frac {4C_2^2 C_{\sigma}} {p^{1+2\sigma}} \right),
\end{align*}
where
$$
C_{\sigma} = \int_0^1 \gamma^{-\sigma} (1-\gamma)^{-\sigma} d\gamma.
$$
\end{lemma}
\begin{proof}
First since we have $N_0\ge \max\{ \left( \frac {2C_2} {2+\sigma} \right)^{\frac 1 {1+\sigma}},
(4C_2)^{\frac 1 {3+\sigma}} \}$, therefore for $r\ge N_0$,
\begin{align*}
\abs{ \prod_{r<q<p} \left(1+ \frac {\xi_q} {q^3} \right)^r
 -1} \le \frac {2C_2} {2+\sigma} \cdot \frac 1 {r^{1+\sigma}}
 \cdot \left( 1 - \left( \frac r p \right)^{2+\sigma} \right).
\end{align*}
For $1\le r<N_0$, clearly
\begin{align*}
\abs{ \prod_{r<q<p} \left(1+ \frac {\xi_q} {q^3} \right)^r
 -1} \le 1 + C_1 + \frac {2C_1 C_2} {(2+\sigma)N_0^{1+\sigma}}.
\end{align*}
For $p-N_0<p_2<p$, we get
\begin{align*}
\abs{ \prod_{p_2<q<p} ( 1+ \frac {\xi_q} {q^3} )^{p_2}-1 }
\le \frac {2N_0 C_2} {p_2^{2+\sigma}}.
\end{align*}
Now we have
\begin{align*}
\abs{N_p} \le &\norm{\tilde f}_\infty \left(
\sum_{p_1<N_0} \abs{ \prod_{p_1 <q <p} ( 1+ \frac {\xi_q} {q^3})^{p_1}-1}
\cdot \abs{ \prod_{p_2<q<p} ( 1+ \frac {\xi_q} {q^3} )^{p_2}-1 } + \right.\\
& \left.
+ \sum_{N_0\le p_1 \le \frac p 2} \abs{ \prod_{p_1 <q <p} ( 1+ \frac {\xi_q} {q^3})^{p_1}-1}
\cdot \abs{ \prod_{p_2<q<p} ( 1+ \frac {\xi_q} {q^3} )^{p_2}-1 }
\right) \\
\le & \norm{\tilde f}_\infty \left(
\sum_{p_1<N_0} \left( 1 + C_1 + \frac {2C_1C_2} {(2+\sigma)N_0^{1+\sigma}} \right)
\cdot \frac {2N_0C_2} {p_2^{2+\sigma}} +\right.\\
& \left. +\sum_{N_0\le p_1 <\frac p2} \left( \frac {2C_2} {2+\sigma} \right)^2
\frac 1 {p_1^{1+\sigma} p_2^{1+\sigma}} \left( 1 - \left( \frac {p_1 } p \right)^{2+\sigma} \right)
\cdot \left( 1 - \left( \frac {p_2} p \right)^{2+\sigma} \right) \right)\\
& \le \norm{\tilde f}_\infty \left(
\frac {(1+C_1+C_2) 2N_0^2 C_2 } {(p-N_0)^{2+\sigma}}
+ \frac {4C_2^2 C_{\sigma}} {p^{1+2\sigma}} \right).
\end{align*}
\end{proof}
\begin{proof}[Proof of the main theorem]
We begin by observing that $\alpha_f = 1-\sigma(G)$. By assumption \ref{07assump_main} we have
$0<-\sigma(G)<1$. Our inductive assumption is that for $r<p$, we have
$$
\xi_r = h_r + \frac 1 r \sum_{q=2}^{r-1} G(q/r) \xi_q,
$$
and for some positive $\epsilon$, $0<\epsilon<\frac 1 2\min\{1+\sigma(G), -\sigma(G) \}$,
$$
\abs{ h_r } \le C_4 r^{\sigma(G)-\epsilon}.
$$
We shall assume that the inductive assumption is justified for some $p_0$ sufficiently large.
Now assume $p\ge p_0$. Lemma \ref{lem_07_13_13} gives us
$$
\abs{\xi_r } \le K \cdot C_4 r^{\sigma(G)}, \quad\ \forall\ r\le p-1.
$$
At step $p$, define
$$
 \hat \xi_p = -p^2 \left( \left( 1 + \frac {\xi_p} {p^3} \right)^{-p} -1 \right).
$$
Then we have (see section 2)
\begin{align*}
\hat \xi_p = p(R_p^{(1)} -R_{p-1}^{(1)}) - \left( p (Q_p-Q_{p-1}) - \frac p{p-1} \hat \xi_{p-1} \right)
- p (N_p -N_{p-1}).
\end{align*}
By lemma \ref{Lem_07_main_11}
$$
\abs{ p(Q_p -Q_{p-1}) - \frac p {p-1} \hat \xi_{p-1}
+ \frac 1 p \sum_{q=2}^{p-1} G(q/p) \xi_q } \le \frac {C_{\epsilon_1}} {p^{1-\epsilon_1}},
$$
where $C_{\epsilon_1}$ depends only on $\epsilon$, $\epsilon_1$, $\alpha$, $C_4$, $K$ and the function $f$. The
parameter $0<\epsilon_1<1$ can be taken arbitrarily small (but $C_{\epsilon_1}$ will be large). Take
$\epsilon_1<1+\sigma(G)-2\epsilon$, then by lemma \ref{lem_07_13_20} we get
\begin{align*}
& \abs{ \hat \xi_p - \frac 1 p\sum_{q=2}^{p-1} G(q/p) \xi_q } \\
\le & \frac {C_{\epsilon_1}} {p^{1-\epsilon_1}} + p \abs{ N_p ( \xi_2,\xi_3,\cdots, \xi_{p-1})}
+ p \abs{ N_p ( \xi_2,\xi_3,\cdots, \xi_{p-1})} + \frac B p \\
\le & \frac {C_{\epsilon_1}} {p^{1-\epsilon_1}} + {D_2} {p^{2\sigma(G)}} + \frac B p
\le C_5 p^{\sigma(G) -2\epsilon},
\end{align*}
where $D_2$, $B$ and $C_5$ are constants. Now use lemma \ref{lem_07_13_13} again to have that
$$
\abs{ \hat \xi_p } \le K \cdot C_5 p^{\sigma(G)}.
$$
Then by lemma \ref{lem_07_13_40}, we have
\begin{align*}
 \abs{h_p} \le  & \abs{\xi_p-\hat\xi_p}+ \abs{ \hat\xi_p - \frac 1 p \sum_{q=2}^{p-1} G(q/p) \xi_q } \\
\le & \frac {C_6} {p^2} + C_5 p^{\sigma(G) -2\epsilon} \le C_4 p^{\sigma(G)-\epsilon}.
\end{align*}
where the last inequality follows if we take $p_0$ to be sufficiently large. We have proved the inductive hypothesis
for all $p$. Now the desired uniform bounds on $\xi_p$ follows again from the stability lemma \ref{lem_07_13_13}.
The main theorem is proved.
\end{proof}
\begin{remark} \label{07:rem:technical}
A technical assumption is used in the proof of the main theorem. Namely, we need to justify
our inductive assumption up to $p=p_0$. This step is mainly done with the help of numerics.
In the special case $f(\gamma)=6\gamma^2-10\gamma+4$, a good estimate of $p_0$ is about
$p_0\approx 100$. This is enough for our purposes.
\end{remark}
\mysection{appendix}
This appendix is devoted to the elementary estimates needed in previous sections. We note
that because of the technical assumption in section 5, we need to get explicit control of constants.
For this reason we gather all the needed explicit estimates in this appendix.
\begin{lemma}
The following easy elementary inequalities are used:
\begin{itemize}
\item Suppose $x>-1$, then
      $$\log(1+x) \le x.$$
\item Suppose $M>N\ge 1$, then for any ${\alpha>0}$
      $$ \sum_{q=N+1}^M \frac 1 {q^{\alpha+1}} \le \frac 1 {\alpha} \left( \frac 1 {N^\alpha} -
      \frac 1 {M^\alpha} \right).$$
\item If $0\le x\le 1$, then
      $$ e^{\frac x 2} \le 1+x. $$
\item If $-\frac 1 4 \le x \le 0 $, then
      $$ e^{2x} \le 1+x.$$
      equivalently we have
      $$ \log(1+x) \ge 2x.$$
\item If $ \abs x \le \frac 1 2$, then
     $$ e^x \le 1 +x + x^2.$$
\item If $ \abs x \le \frac 1 4$, then
     $$ \log(1+x) \ge x -x^2.$$
\end{itemize}
\end{lemma}
\begin{lemma} \label{lem_07_13_40}
Let $D_1>0$, $p_0$ is an integer and $p_0\ge \max\{ 2\sqrt{D_1}, 7\}$.
Suppose $\abs x \le D_1$, Then
$$
\abs{ p^3\left( \left(1+\frac x {p^2}\right)^{-\frac 1 p} -1 \right) +x} \le
\frac {5x^2} {4p^2} \le \frac {5D_1^2} {4p^2}, \quad \forall\ p\ge p_0.
$$
\end{lemma}
\begin{proof}
One side is trivial:
\begin{align*}
& p^3\left( \left(1+\frac x {p^2}\right)^{-\frac 1 p} -1 \right) +x \\
\ge & p^3 \left( \exp\left( -\frac x {p^3} \right) -1\right) +x
\ge 0.
\end{align*}
For the other side,
\begin{align*}
& p^3\left( \left(1+\frac x {p^2}\right)^{-\frac 1 p} -1 \right) +x \\
\le & p^3 \left( \exp\left( -\frac x{p^3} + \frac {x^2} {p^5} \right)-1 \right) +x \\
\le & p^3 \left( - \frac x{p^3} + \frac {x^2} {p^5} +
  \left( \frac x{p^3} -\frac {x^2} {p^5} \right)^2 \right) +x \\
\le & \frac {5x^2} {4p^2}.
\end{align*}
\end{proof}
\begin{lemma}
Assume that $\abs{\xi_q} \le A$, for any $q\ge N_0$, where $N_0\ge \max\{1, A\}$. Then we have
$$
\abs{ \prod_{p_1<q<p} \left(1+ \frac{\xi_q} {q^3} \right)^{p_1} -1} \le \frac A {p_1},
$$
for any $p_1 \ge N_0$.
\end{lemma}
\begin{proof}
We have by assumption $\frac {\abs{\xi_q}} {q^3} \le \frac 1 4$, therefore
\begin{align*}
\prod_{p_1<q<p} \left( 1+ \frac {\xi_q} {q^3} \right)^{p_1}
 & = \exp \left\{ \sum_{q=p_1+1}^{p-1} p_1 \log \left( 1+ \frac {\xi_q} {q^3} \right) \right\}\\
 & \le \exp\left\{ \sum_{q=p_1+1}^{p-1} p_1 \frac {A} {q^3} \right\} \\
 & \le \exp\left\{ \frac A {2p_1} \right\} \le 1+ \frac A {p_1}.
\end{align*}
Similarly,
\begin{align*}
\prod_{p_1<q<p} \left( 1+ \frac {\xi_q} {q^3} \right)^{p_1}
& = \exp \left\{ \sum_{q=p_1+1}^{p-1} p_1 \log \left( 1+ \frac {\xi_q} {q^3} \right) \right\}\\
& \ge \exp \left\{ \sum_{q=p_1+1}^{p-1} p_1 \log \left( 1 - \frac A {q^3} \right) \right\}\\
& \ge \exp \left\{ \sum_{q=p_1+1}^{p-1} p_1 ( - \frac {2A} {q^3} ) \right\}\\
& \ge \exp \left\{ - \frac A {p_1} \right\} \ge 1 - \frac A {p_1}.
\end{align*}
\end{proof}
\begin{lemma}
$$
\sum_{N_0 \le p_1 \le p-3} \frac {\gamma^2} {(p-1)^2}
\abs{ \prod_{p_1 <q < p} \left(1+\frac {\xi_q} {q^3} \right)^{p_1} -1 }
\le \frac A {2p^2}.
$$
\end{lemma}
\begin{proof}
we have
\begin{align*}
&\sum_{N_0 \le p_1 \le p-3} \frac {\gamma^2} {(p-1)^2}
\abs{ \prod_{p_1 <q < p} \left(1+\frac {\xi_q} {q^3} \right)^{p_1} -1 } \\
\le & \sum_{N_0 \le p_1 \le p-3} \frac {p_1^2} {(p-1)^2 p^2} \cdot \frac A {p_1} \\
\le & \frac {A} {2p^2}.
\end{align*}
\end{proof}
\begin{lemma} \label{lem_07_16_34}
Assume that $\abs{ \xi_q } \le A_1$, for any $q\ge N_0$, where $N_0 \ge \max\{1,A_1\}$.
Also assume that for some constant $A_2 >0$,
$$
\prod_{p_1 < q\le N_0} \left(1+ \frac {\xi_q} {q^3} \right)^{p_1} \le A_2, \quad
\forall\ 1\le p_1 \le N_0.
$$
Then we have
$$
\sum_{1\le p_1 \le p-3} \frac {\gamma^2} {(p-1)^2}
\abs{ \prod_{p_1<q<p} \left( 1+ \frac {\xi_q} {q^3} \right)^{p_1} -1}
\le \frac {A_1} {2p^2} + \frac {N_0^3} {3(p-1)^2 p^2} \left( 1 +A_2(1+ \frac {A_1}{N_0}) \right).
$$
\end{lemma}
\begin{proof}
We have
\begin{align*}
&\sum_{p_1 < N_0} \frac {\gamma^2}
{(p-1)^2} \abs{ \prod_{p_1 <q <p} \left(1+\frac {\xi_q}{q^3} \right)^{p_1} -1} \\
\le & \sum_{p_1<N_0} \frac {p_1^2} {(p-1)^2 p^2}
 \left( A_2 \prod_{N_0 <q <p} \left( 1+ \frac {A_1} {q^3} \right)^{N_0} +1 \right) \\
\le & \sum_{p_1 < N_0} \frac {p_1^2} {(p-1)^2 p^2} \left(1 + A_2 ( 1+ \frac {A_1}{N_0} )\right)\\
\le & \frac {N_0^3} {3(p-1)^2 p^2} \left(1 + A_2 ( 1+ \frac {A_1}{N_0} )\right).
\end{align*}
\end{proof}
\begin{lemma}
Assume that $\abs{\xi_q} \le A$, for any $q\ge N_0$, where $N_0\ge \max\{1, A\}$. Then we have
$$
\abs{ \prod_{p_1<q<p} \left(1+ \frac{\xi_q} {q^3} \right)^{p_1} - \left( 1 +
 p_1\sum_{p_1<q<p} \frac {\xi_q} {q^3} \right) } \le \frac {A^2} {4p_1^2},
$$
for any $p_1 \ge N_0$.
\end{lemma}
\begin{proof}
First we have
\begin{align*}
\abs{ \sum_{p_1 <q <p} \frac {p_1 \xi_q } {q^3} }
\le \sum_{p_1<q<p} \frac {p_1 A} {q^3}
\le \frac {A} {2p_1} \le \frac 1 2.
\end{align*}
Therefore
\begin{align*}
\prod_{p_1 <q <p} \left( 1 + \frac {\xi_q} {q^3} \right)^{p_1}
& = \exp\left\{ p_1 \sum_{p_1<q<p} \log \left(1+\frac {\xi_q}{q^3} \right)
      \right\} \\
& \le \exp\left\{ \sum_{p_1<q <p} \frac {p_1 \xi_q} {q^3} \right\} \\
& \le 1+ \sum_{p_1<q<p} \frac {p_1 \xi_q} {q^3}
  + \frac {A^2} {4p_1^2}.
\end{align*}
Similarly
\begin{align*}
&\prod_{p_1 <q <p} \left( 1 + \frac {\xi_q} {q^3} \right)^{p_1}
 - \left( 1 + p_1 \sum_{p_1<q<p} \frac {\xi_q} {q^3} \right) \\
\ge & \exp\left\{ p_1 \sum_{p_1<q<p} \left(\frac {\xi_q} {q^3} -
      \frac {\xi_q^2} {q^6} \right) \right\}
-\left( 1 + p_1 \sum_{p_1<q<p} \frac {\xi_q} {q^3} \right) \\
\ge & \frac 3 2 \left( e^{-p_1 \sum_{p_1<q<p} \frac {A^2} {q^6} } -1 \right)\\
\ge & \frac 3 2 \left( e^{-\frac {A^2} {5p_1^4} } -1 \right) \\
\ge & -\frac {3A^2} {10p_1^4} \ge - \frac {A^2} {4p_1^2}.
\end{align*}
\end{proof}
\begin{lemma} \label{lem_07_16_39}
Assume the following:
\begin{enumerate}
\item $\abs{\xi_q} \le A_1$, for any $q\ge N_0$, where $N_0 \ge \max\{1,A_1\}$.
\item $\abs{\xi_q} \le A_3$ for any $1\le q\le N_0$.
\item For some constant $A_2 >0$,
$$
\prod_{p_1 < q\le N_0} \left(1+ \frac {\xi_q} {q^3} \right)^{p_1} \le A_2, \quad
\forall\ 1\le p_1 \le N_0.
$$
\end{enumerate}
Then we have
\begin{align*}
&\sum_{1 \le p_1 \le p-3} \frac {\gamma} {p-1}
 \abs{ \prod_{p_1<q <p} \left( 1+ \frac {\xi_q} {q^3} \right)^{p_1} -
  \left( 1 + p_1\sum_{p_1<q<p} \frac {\xi_q} {q^3} \right) } \\
\le & \frac {A_1^2 \log (p-3) } {4p(p-1)} +
\frac 1 {p(p-1)} \left(
 \frac {1+A_2} 2 N_0^2 + \frac {(3A_2+1)A_1 +3 A_3} 6 N_0 \right).
\end{align*}
\end{lemma}
\begin{proof}
We have
\begin{align*}
&\sum_{N_0 \le p_1 \le p-3} \frac {\gamma} {p-1}
 \abs{ \prod_{p_1<q <p} \left( 1+ \frac {\xi_q} {q^3} \right)^{p_1} -
  \left( 1 + p_1\sum_{p_1<q<p} \frac {\xi_q} {q^3} \right) } \\
\le & \sum_{N_0\le p_1\le p-3} \frac {p_1} {p(p-1)} \cdot \frac {A_1^2} {4p_1^2} \\
\le & \frac {A_1^2 \log (p-3) } {4p(p-1)}.
\end{align*}
Also
\begin{align*}
\abs{ p_1 \sum_{p_1<q<p} \frac {\xi_q} {q^3} }
& \le p_1 \sum_{p_1<q\le N_0} \frac {A_3} {q^3}
+ p_1 \sum_{N_0 <q <p} \frac {A_1} {q^3} \\
& \le \frac {A_3} {2p_1} + \frac {p_1 A_1} {2N_0^2}.
\end{align*}
This gives us
\begin{align*}
&\sum_{1 \le p_1 < N_0} \frac {\gamma} {p-1}
 \abs{ \prod_{p_1<q <p} \left( 1+ \frac {\xi_q} {q^3} \right)^{p_1} -
  \left( 1 + p_1\sum_{p_1<q<p} \frac {\xi_q} {q^3} \right) } \\
\le & \frac 1 {p(p-1)}
 \sum_{1\le p_1 <N_0} p_1 \cdot
 \left( \abs{ \prod_{p_1<q<p} \left( 1+ \frac {\xi_q} {q^3} \right)^{p_1} -1}
  + \abs{ p_1\sum_{p_1<q<p} \frac {\xi_q} {q^3} } \right) \\
\le & \frac 1 {p(p-1)} \sum_{1\le p_1 <N_0} p_1
\cdot \left( 1 + A_2 ( 1+ \frac {A_1}{N_0} )
 + \frac {A_3} {2p_1} + \frac {p_1 A_1} {2N_0^2} \right)\\
 \le & \frac 1 {p(p-1)} \left(
 \frac {1+A_2} 2 N_0^2 + \frac {(3A_2+1)A_1 +3 A_3} 6 N_0 \right).
\end{align*}
\end{proof}
\section*{Acknowledgment}
The author would like to thank Jean Bourgain, Tom Spencer and Ya G. Sinai for
their interests in this work.

\end{document}